\documentclass[journal,onecolumn,12pt]{IEEEtran}
\usepackage[numbers,sort&compress]{natbib}
\usepackage{pslatex}
\usepackage{amsfonts,color,morefloats}
\usepackage{amssymb,amsmath,latexsym,amsthm}
\usepackage{hyperref}
\hypersetup{hidelinks}

\newcommand{\gf}{{\mathbb{F}}}

\newtheorem{theorem}{Theorem}
\newtheorem{lemma}[theorem]{Lemma}

\newtheorem{corollary}[theorem]{Corollary}

\newtheorem{definition}{Definition}
\newtheorem{example}{Example}

\begin{document}
	\title{ Two low differentially uniform power permutations over odd characteristic finite fields: APN and differentially $4$-uniform functions}
	\author{\IEEEauthorblockN{Haode Yan\IEEEauthorrefmark{1},
	Sihem Mesnager\IEEEauthorrefmark{2}, and Xiantong Tan\IEEEauthorrefmark{1}}\\
	\IEEEauthorblockA{\IEEEauthorrefmark{1}School of Mathematics, Southwest Jiaotong University, Chengdu, China. }\\
	\IEEEauthorblockA{\IEEEauthorrefmark{2}Department of Mathematics, University of Paris VIII, 93526 SaintDenis, with University Sorbonne Paris Cit{\'{e}}, LAGA, UMR 7539, CNRS, 93430 Villetaneuse, and also with the T{\'{e}}l{\'{e}}com Paris, France. }\\
	\IEEEauthorblockA{\href{mailto: hdyan@swjtu.edu.cn}{hdyan}@swjtu.edu.cn, \href{mailto: smesnager@univ-paris8.fr}smesnager@univ-paris8.fr}, \href{mailto: Xiantong@my.swjtu.edu.cn}{Xiantong}@my.swjtu.edu.cn\\
	\IEEEauthorblockA{Corresponding Author: Sihem Mesnager \quad Email: smesnager@univ-paris8.fr}}
	\maketitle
	\begin{abstract}

	Permutation polynomials over finite fields are fundamental objects as they are used in various theoretical and practical applications in cryptography, coding theory, combinatorial design, and related topics. This family of polynomials constitutes an active research area in which advances are being made constantly. In particular, constructing infinite classes of permutation polynomials over finite fields with good differential properties (namely, low) remains an exciting problem despite much research in this direction for many years.

 This article exhibits low differentially uniform power permutations over finite fields of odd characteristic. Specifically, its objective is twofold concerning the power functions $F(x)=x^{\frac{p^n+3}{2}}$ defined over the finite field $\gf_{p^n}$ of order $p^n$, where $p$ is an odd prime, and $n$ is a positive integer. The first is to complement some former results initiated by Helleseth and Sandberg in \cite{HS} by solving the open problem left open for more than twenty years concerning the determination of the differential spectrum of $F$ when $p^n\equiv3\pmod 4$ and $p\neq 3$. The second is to determine the exact value of its differential uniformity. Our achievements are obtained firstly by evaluating some exponential sums over $\gf_{p^n}$ (which amounts to evaluating the number of $\gf_{p^{n}}$-rational points on some related curves and secondly by computing the number of solutions in $(\gf_{p^n})^4$ of a system of equations presented by Helleseth, Rong, and Sandberg in [``New families of almost perfect nonlinear power mappings," IEEE Trans. Inform. Theory, vol. 45. no. 2, 1999], naturally appears while determining the differential spectrum of $F$. We show that in the considered case ($p^n\equiv3\pmod 4$ and $p\neq 3$), $F$ is an APN power permutation when $p^n=11$, and a differentially $4$-uniform power permutation otherwise.

	\end{abstract}
	
	{\bf Keywords:} Vectorial function,  Power function, APN,  Differential uniformity, Differential spectrum, Elliptic curve, Permutation, Finite field.
	
	{\bf Mathematics Subject Classification:} 06E30, 11T06, 94A60, 94D10.
	
	\section{Introduction}

	Let $\gf_{q}$ be the finite field with $q$ elements, where $q$ is a prime power ($q=p^n$) and $n$ is a positive integer. We denote by $\gf_{q}^*$ the multiplicative cyclic group of nonzero elements of the finite field $\gf_{p^n}$. Any function $F: \gf_{q} \rightarrow \gf_{q}$ from $\gf_{q}$ to itself can be uniquely represented as a univariate polynomial of degree less than $q$. Therefore, $F$ can always be seen as a polynomial in $\gf_{q}[x]$. A polynomial $F \in  \gf_{q}[x]$ is called a permutation polynomial (PP) of $ \gf_{q}$ if the mapping $x \mapsto F(x)$ is a permutation of $ \gf_{q}$. Permutation polynomials over finite fields are important objects not only theoretically but also in practical applications such as cryptography and related topics. Nowadays,  designing infinite classes of permutation polynomials over finite fields with good cryptographic properties remains an exciting research topic.
	
In the binary case ($p=2$), such functions $F$ (also called vectorial Boolean functions) are important components in symmetric cryptography. When they are used inside a symmetric cryptosystem (namely, in a block cipher), they are called Substitution boxes (S-boxes for short) since they play a crucial role in the security of such ciphers. In this cryptographic context, cryptographic functions should resist differential cryptanalysis, introduced by Biham and Shamir~\cite{BS}, which is one of the most powerful attacks on block ciphers. To quantify the ability of a given function $F$ to resist the differential attack, Nyberg~\cite{Nyberg93} introduced the notion of differential uniformity, which is closely related to its difference distribution table (DDT for short). These tools are defined as follows.
 For any $ a, b \in \gf_{q}$, the DDT entry at point $ (a,b) $, denoted by $ \delta_F(a,b) $, is defined as \[\delta_F(a,b)=\big{|} \{x \in \gf_{q}: ~F(x+a)-F(x)=b\} \big{| },\]
	where $ \big{|} S \big{|} $ denotes the cardinality of the set $S$.
	The differential uniformity of the function $ F $, denoted by $ {\Delta _F} $, is then defined as
	\[{\Delta _F}=\max\{\delta_F(a,b): ~a \in {\mathbb{F}_{q}^*}, b \in \mathbb{F}_{q}\}  ,\]
	where $\gf_{q}^*=\gf_{q}\setminus\{0\}$. For an S-box $F$, the smaller the value $\Delta_F$ is, the better the contribution of $ F $ to the resistance against differential attack. When $\Delta_{F}=1$, $F$ is said to be \emph{perfect nonlinear} (PN for short) function. 
	Whereas, when $\Delta_F=2$, $F$ is called an \emph{almost perfect nonlinear} (APN for short) function. Note that PN functions over even characteristic finite fields do not exist. PN and APN functions play an important role in the theoretical aspects of several domains.   
	 Particularly, vectorial functions $F$, which are permutations, have an extreme interest, especially in symmetric cryptography. Typically, block ciphers use a permutation as an S-box during the encryption process, while the compositional inverse of the S-box is used during the decryption process. A nice survey on cryptographic vectorial Boolean functions (including developments on differential uniform functions) can be found in the recent book \cite{CarletBook}. Recent progress on cryptographic functions with low differential uniformity can be found in \cite{BD,BCHLS,BCL,CM,DO,DOB1,DOB2,HRS,HS,PT,QTLG,QTTL,TCT,TZ,ZHS,ZKW,ZW} and the references therein.
	
		Power functions with low differential uniformity serve as good candidates for the
	design of S-boxes because of their strong resistance to differential attacks and the usually low implementation cost in hardware. When $F$ is a power function, i.e., $F(x)=x^d$ for an integer $d$, one easily see that $\delta_F(a,b)=\delta_F(1,{b/{a^d}})$ for all $a\in \gf_{q}^*$ and $b\in \gf_{q}$.
	That is to say, the differential properties of $F$ are wholly determined by the values of $\delta_F(1,b)$ as $b$ runs through $\gf_{q}$. Their resistance to the standard differential attack attracted attention and was investigated. The notion of the differential spectrum of a power function was, in fact, firstly proposed by Blondeau, Canteaut, and Charpin in \cite{BCC} as follows.
	
	\begin{definition}\label{def1}
		Let $F(x)=x^d$ be a power function over $\gf_{q}$ with differential uniformity $ \Delta_F $. Denote
		\[\omega_i= \big{|} \left\{b\in \gf_{q}: \delta_F(1, b)=i\right\} \big{|} ,\,\,0\leq i\leq \Delta_F.\]
		The differential spectrum of $F$ is defined by the following multi-set
		\[ DS_F =\left\{\omega_i>0:0\leq i \leq \Delta_F \right\}.\]
	\end{definition} 

The  DDT distribution of a power function can be deduced via its differential spectrum. Moreover, it has been shown in \cite{BCC} that the elements in the differential spectrum of $F$ satisfy the two following helpful identities in the set $\mathbb{N}$ of  the positive integers.
	
	\begin{equation}\label{omegaiomega}
		\sum_{i=0}^{\Delta _{F}}\omega _{i}=\sum_{i=0}^{\Delta _{F}}(i\cdot \omega _{i})=q.
	\end{equation}
	The following lemma plays an important role
	in determining the differential spectrum of $F$.
	\begin{lemma}\label{identicalequation}
		(\cite{HRS}, Theorem 10)
			Keep the notation introduced in Definition \ref{def1}. Denote by $N_{4}$ the
		number of solutions $(x_{1},x_{2},x_{3},x_{4})\in (\gf_{q})^{4}$ of the equation system
		\begin{equation}\label{equsystem}
			\Bigg\{ \begin{array}{ll}
				{x_1} - {x_2} + {x_3} - {x_4} &= 0\\
				x_1^d - x_2^d + x_3^d - x_4^d &= 0.
			\end{array}
		\end{equation}
		Then we have 
		\begin{equation}\label{i^2omega}
			\sum_{i=0}^{\Delta _{F}}i^{2}\omega_{i}=\frac{N_{4}-q^2}{q-1}.
		\end{equation}
	\end{lemma}
	
Even though determining the differential spectrum of functions and designing whose which are bijective of low differential uniformity is of high importance in the cryptographic framework, estimating the differential behaviors of functions and designing low differential uniform ones (such as APN, PN, differential $4$-uniform functions) in the non-binary case ($p$ odd), as demonstrated in many precious papers in the literature, also remains a high interest. However, it is challenging to determine a power function's differential spectrum ultimately. Only a few classes of power functions over odd characteristic finite fields have known differential spectra. The known results on power functions $F$ over $\gf_{p^{n}}$ ($p$ is odd) for which differential spectrum has been determined, are summarized in Table \ref{t1}. 
	
		\begin{table}[h]
		\renewcommand{\arraystretch}{1.7}
		\caption{Power functions $F(x)=x^{d}$ over $\gf_{p^{n}}$ ($p$ is odd) with known differential spectrum}
		\label{t1}
		\centering
		\begin{tabular}{|c|c|c|c|}
			\hline
			$d$ & Conditions & $\Delta_{F}$ & Ref. \\ 
			\hline\hline
			$2\cdot 3^{\frac{n-1}2}+1$ & $n$ is odd & 4 & \cite{Dobbertin2001}\\
			\hline
			$\frac{p^k+1}2$ & $\gcd(n,k)=e$ & $\frac{p^e-1}2$ or $p^e+1$ & \cite{CHNC}\\
			\hline
			$\frac{p^n+1}{p^m+1}+\frac{p^n-1}{2}$ & $p\equiv3\pmod4$, $m\mid n$ and $n$ odd & $\frac{p^m+1}2$  & \cite{CHNC}\\
			\hline
			$p^{2k}-p^k+1$ & $\gcd(n,k)=e$,\ $\frac ne$ is odd & $p^e+1$ & \cite{Lei2021,Yan}\\
			\hline
			$p^n-3$ & any $n$ & $\leq$5 & \cite{XZLH,YXLHXL}\\
			\hline
			$p^m+2$ &  $p>3$, $n=2m$ & 2 or 4  & \cite{MXLH}\\
			\hline
			$\frac{5^n-3}2$ & any $n$ & 4 or 5 & \cite{YL}\\\hline
			$\frac{p^n-3}2$ & $p^n\equiv3\pmod4$, $p^n>7$ and $p^{n}\ne 27$ & 2 or 3 & \cite{YMT}\\\hline
			$\frac{p^n+3}2$ & $p\geq 5$, $p^n\equiv1\pmod4$ & 3 & \cite{JLLQ}\\\hline
				$\frac{p^n+3}2$ & $p^n= 11$ & 2 & This paper\\\hline
			$\frac{p^n+3}2$ & $p^n\equiv3\pmod4$, $p\neq 3$, $p^n\neq 11$ & 4 & This paper\\\hline
		\end{tabular}
	\end{table}

	In this paper, we shall focus on the power function $F(x)=x^{\frac{p^n+3}{2}}$ over $\gf_{p^n}$, where $p$ is an odd prime and $n$ is a positive integer. 
	The differential properties of such $F$ have attracted interest since 1997 when Helleseth and Sandberg studied the differential uniformity of $F$ and proved the following main result.
	\begin{theorem}[\cite{HS}, Theorem 3]\label{tor}Let $p$ be an odd prime, $d=\frac{p^n+3}{2}$ and let $F(x)=x^d$, then
		\begin{equation*}
			\Delta_F \leq \left\{\begin{array}{ll}
				1, & ~\mathrm{if}~p=3~\mathrm{and}~n~\mathrm{even},\\ 
				3, & ~\mathrm{if}~p\neq3~\mathrm{and}~p^n\equiv1~(\mathrm{mod}~4),\\
				4, & ~\mathrm{otherwise.}
			\end{array}\right.
		\end{equation*}
	\end{theorem}
	
		It is known that  if $p=3$ then $F$ it is equivalent to the power function $x^{\frac{3^{n-1}+1}{2}}$,  whose differential spectrum was determined in \cite{CHNC}.  
The differential spectrum of $F$  in the case $p^n\equiv1\pmod 4$ and $p\neq3$  was determined by  Jiang,  Li,  Li, and  Qu (\cite{JLLQ}) by proving that $F$  is a differential $3$-uniform.  Very recently, the same authors have resolved in  \cite{JLLQ2} the problem of determining the differential spectrum of $F$ when $p=3$ and $n$ is an odd integer. However, the problem of determining the differential spectrum and its corresponding differential uniformity is left open in the remaining case where  $p^n\equiv3\pmod 4$ and $p\neq 3$. The ultimate objective of this article is to study the differential spectrum of $F$ in the open case. We emphasize  that  power functions $F(x)=x^{\frac{p^n+3}{2}}$ over $\gf_{p^n}$ where  $p^n\equiv3\pmod 4$ and $p\neq 3$, are  very interesting since they induce PP. We shall employ several mathematical tools to  prove that these PP process a low differential uniformity, precisely either $2$ (that is APN functions) or $4$.

The rest of this paper is organized as follows. In Section \ref{A}, we  recall some preliminaries,  introduce some basic concepts related to quadratic multiplicative characters,  and present results on quadratic multiplicative character sums,  which will be useful in the rest of the paper. We shall notice that mathematically, the determination of the differential spectrum of $F$ and the computation of its corresponding differential uniformity amounts to determine the number of solutions of certain related equation systems over $(\gf_{p^n})^4$ presented by Helleseth et al. in \cite{HRS} and to evaluate some related exponential sums over $\gf_{p^n}$ (which amounts to evaluate the number of $\gf_{p^{n}}$-rational points on some related curves). The results of our approach are given in Section \ref{B}, which particularly gives the number of solutions of the considered equation systems. Using our approach and the derived mathematical results in Section \ref{B}, we present in Section \ref{C} our computation of the differential spectrum of $F$ and, next, the corresponding differential uniformity. Finally, Section \ref{E} concludes this paper.

	\section{On quadratic character sums}\label{A}
Let $\gf_{p^n}$ be the finite field with $p^n$ elements, where $p$ is an odd prime, and $n$ is a positive integer. And let $\gf_{p^n}^{*}=\gf_{p^n}\setminus\{0\}$. 
This section mainly introduces some basic results on quadratic multiplicative character sums $\chi$ over $\gf_{p^n}$, i.e.,
\begin{equation*}
	\chi(x)=x^{\frac{p^{n}-1}{2}}=\left\{\begin{array}{ll}
		1, & \hbox{if $x$ is a square,}\\ 
		0, & \hbox{if $x=0$,}\\
		-1, & \hbox{if $x$ is a nonsquare.}
	\end{array}\right.
\end{equation*}
These exponential sums will appear naturally in our study of the differential spectra of the function $F(x)=x^{\frac{p^{n}+3}{2}}$ over the finite field $\gf_{p^n}$. These connections are highlighted below with some results in the case where $p^n \equiv 3 (\mathrm{mod}~4)$.\\

Let $\gf_{{p^{n}}}[x]$ be the polynomial ring over $\gf_{{p^{n}}}$. We consider the sum involving the quadratic multiplicative sums of the form
	\begin{equation*}
		\sum_{x\in \gf_{{p^{n}}}}\chi(f(x))
	\end{equation*}
	with $f(x)\in \gf_{{p^{n}}}[x]$. The case of $\deg(f(x))=1$ is trivial, and for $\deg(f(x))=2$, the following explicit formula was established in \cite{FF}.
	
	\begin{lemma}\label{2}
		(\cite{FF}, Theorem 5.48)
		Let $f(x)=a_{2}x^{2}+a_{1}x+a_{0}\in \gf_{{p^{n}}}[x]$ with $p$ odd and $a_{2} \neq 0$.  Set $\Delta:=a_{1}^{2}-4a_{0}a_{2}$ the discriminant of $f(x)=0$ and denote by $\chi$ be the quadratic character of $\gf_{{p^{n}}}$. Then
		\begin{equation*}
			\sum_{x\in \gf_{{p^{n}}}}\chi(f(x))=\left\{\begin{array}{ll}
				-\chi(a_{2}), & ~\mathrm{if}~\Delta\neq0,\\ 
				(p^{n}-1)\chi(a_{2}),  & ~\mathrm{if}~\Delta=0.
			\end{array}\right.
		\end{equation*}
	\end{lemma}
For $\deg(f(x))\geq 3$, it is challenging to derive an explicit a general formula for the character sum $\sum\limits_{x\in \gf_{p^{n}}}\chi(f(x))$. However, when $\deg(f(x))=3$, such a sum can be computed by considering $\gf_{p^{n}}$-rational points of elliptic curves over $\gf_{p}$. More specifically, for a cubic function $f$, we denote $\Gamma _{p,n}$ as
\begin{equation}
	\label{Gamma}
	\Gamma _{p,n}=\sum_{x\in \gf_{p^{n}}}\chi(f(x)).
\end{equation}

To evaluate  $\Gamma _{p,n}$, we shall use some elementary concepts from the theory of elliptic curves. Most of the terminologies and notation are  borrowed from \cite{SilveEC}. Let $E/\gf_{p}$ be the elliptic curve over $\gf_{p}$:
\begin{equation*}
	E:y^{2}=f(x).
\end{equation*}
Let $N_{p,n}$ denote the number of $\gf_{p^{n}}$-rational points (remember the extra point at infinity) on the curve $E/\gf_{p}$. From Subsection 1.3 in [\cite{SilveEC}, P. 139, Chap. V] and Theorem 2.3.1 in [\cite{SilveEC}, P. 142, Chap. V], $N_{p,n}$  can be computed from $\Gamma _{p,n}$. More precisely,  for every $n\geq1$,
\begin{equation*}
	N_{p,n}=p^n+1+\Gamma _{p,n}.
\end{equation*}
Moreover, 
\begin{equation}\label{gamma}
	\Gamma _{p,n}=-\alpha^{n}-\beta^{n},
\end{equation}
where $\alpha$ and $\beta$ are the complex solutions of the quadratic equation $T^{2}+\Gamma_{p,1}T+p=0$.

We  are now interested in two specific character sums $\lambda^{(1)}_{p,n}$ and $\lambda^{(2)}_{p,n}$. Let
	\begin{align}\label{sum1}
		\lambda^{(1)}_{p,n}&=\sum_{x\in \gf_{p^{n}}}\chi(x(x+1)(x-3))
	\end{align}
	and
	\begin{align}\label{sum2}
		\lambda^{(2)}_{p,n}&=\sum_{x\in \gf_{p^{n}}}\chi(x(x+1)(x-2)).
	\end{align}

 The former sums will be helpful in Section \ref{B}  since, as we shall see, the computation of the differential spectrum of  $F(x)=x^{\frac{p^{n}+3}{2}}$ over $\gf_{{p^{n}}}$ boils down to evaluating those character sums.

In the following examples, we give the exact values of $\lambda^{(1)}_{p,n}$ and $\lambda^{(2)}_{p,n}$, 
respectively,  over prime fields (for specific values of $p$).
\begin{example}
Let $p=7$. For $n=1$, one has $\lambda^{(1)}_{7,1}=0$. The quadratic equation $T^{2}+7=0$ has two complex roots $\pm\sqrt{-7}$. By \eqref{gamma},  we have
		\begin{align*}
			\lambda^{(1)}_{7,n}&=-(\sqrt{-7})^{n}-(-\sqrt{-7})^{n}\\&=\left\{\begin{array}{ll}
				(-1)^{\frac{n}{2}+1}\cdot 2\cdot {7}^{\frac{n}{2}}, & ~\mathrm{n}~ \mathrm{is}~\mathrm{even},\\
				0, & ~\mathrm{n}~ \mathrm{is}~\mathrm{odd},
			\end{array}\right.
		\end{align*}
		Moreover,  when $n=1$, $\lambda^{(2)}_{7,1}=-4$. The quadratic equation $T^{2}-4T+7=0$ has two complex roots $2\pm \sqrt{-3}$. Similarly, we have 
		\begin{align*}
			\lambda^{(2)}_{7,n}&=-(2+\sqrt{-3})^{n}-(2-\sqrt{-3})^{n}\\&=\sum\limits_{k = 0}^{\left\lfloor {\frac{n}{2}} \right\rfloor } {{{\left( { - 1} \right)}^{k+1}} \binom{n}{2k} {2^{n-2k+1}}\cdot{3^{k}}}.
		\end{align*}
\end{example}
In addition, we have the following bound on $\lambda^{(i)}_{p,n}$ for $i=1,2$.
	\begin{theorem}[\cite{SilveEC}, Corollary 1.4]
		\label{bound}
		Keep the notation as above. Then we have  $|\lambda^{(i)}_{p,n}|\leq 2p^{\frac{n}{2}}$,
		for $i=1,2$.
	\end{theorem}
We present the following results concerning the exact values of six specific character sums used in Section \ref{B}.
	\begin{lemma}\label{sums}
		Let $p^n \equiv 3 (\mathrm{mod}~4)$, we have
		\begin{align*}
			&1)\sum_{x\in \gf_{p^{n}}}\chi\big{(}x(x^2+x+1)\big{)}=\lambda^{(1)}_{p,n},\\
			&2) \sum_{x\in \gf_{p^{n}}}\chi\big{(}(x+1)(x^2+x+1)\big{)}=-\lambda^{(1)}_{p,n},\\
			&3)\sum_{x\in \gf_{{p^{n}}}}\chi\big{(}(x^2+x)(x^2+x+1)\big{)}=-\lambda^{(1)}_{p,n}-1.
		\end{align*}
		Moreover, when $p^n \equiv 3 (\mathrm{mod}~4)$ and $p\neq 3$, we have
		\begin{align*}
			&4)\sum_{x\in \gf_{p^{n}}}\chi\big{(}x(3x^2+2x+3)\big{)}=\lambda^{(2)}_{p,n},\\
			&5)\sum_{x\in \gf_{{p^{n}}}}\chi\big{(}(x^2+x+1)(3x^2+2x+3)\big{)}=\lambda^{(2)}_{p,n}-\chi(3),\\
			&6)\sum_{x\in \gf_{p^{n}}}\chi\big{(}x(x^2+x+1)(3x^2+2x+3)\big{)}=-2\lambda^{(1)}_{p,n},
		\end{align*}
	where $\lambda^{(1)}_{p,n}$ and $\lambda^{(2)}_{p,n}$ are defined in \eqref{sum1} and  \eqref{sum2}, respectively.
	\end{lemma}
	\begin{proof}
	\begin{enumerate}
		\item We have
		\begin{align*}
			\sum_{x\in \gf_{p^{n}}}\chi\big{(}x(x^2+x+1)\big{)}&=\sum_{x\in\gf_{p^n}^* }\chi\big{(}x(x^2+x+1)\big{)}\\&=\sum_{x\in\gf_{p^n}^* }\chi\big{(}\frac{x^2+x+1}{x}\big{)}.
		\end{align*}
		Let $\frac{x^2+x+1}{x}=u$, then $x$ and $u$ satisfy
		\begin{equation}\label{u-xrelation1}
			x^2+(1-u)x+1=0,
		\end{equation}
		which is a quadratic equation on $x$ with the discriminant $\Delta=(u+1)(u-3)$. For each $u\in\gf_{{p^{n}}}$, it corresponds $(1+\chi(\Delta))$ $x$'s from \eqref{u-xrelation1}. Then we obtain
		\begin{align*}
			\sum_{x\in\gf_{p^n}^*}\chi\big{(}\frac{x^2+x+1}{x}\big{)}&=\sum_{u\in\gf_{p^n}}\chi(u)\big{(}1+\chi((u+1)(u-3))\big{)}\\&=\sum_{u\in\gf_{p^n}}\chi\big{(}u(u+1)(u-3)\big{)}\\
			&=\lambda^{(1)}_{p,n}.
		\end{align*}
		The desired result follows.
		
		\item Note that $-x-1$ permutes $\gf_{{p^n}}$ and set $y=-x-1$. Then 
		\begin{align*}
			\sum_{x\in \gf_{p^{n}}}\chi\big{(}(x+1)(x^2+x+1)\big{)}&=\sum_{y\in\gf_{p^n}}\chi\big{(}(-y)(y^{2}+y+1)\big{)}\\&=\chi(-1)\sum_{y\in\gf_{p^n}}\chi\big{(}y(y^{2}+y+1)\big{)}\\&=-\lambda^{(1)}_{p,n}
		\end{align*}
		since $\chi(-1)=-1$. 
		
		
		\item We consider the character sum $\sum\limits_{x\in \gf_{{p^{n}}}}\chi\big{(}(x^2+x)(x^2+x+1)\big{)}$.  Set $x^2+x=u$. Then $x$ and $u$ satisfy 
		\begin{equation}\label{u-xrelation3}
			x^2+x-u=0,
		\end{equation}
		which is a quadratic equation on $x$ with the discriminant $\Delta=4u+1$. For each $u\in\gf_{{p^{n}}}$, it corresponds $\big{(}1+\chi(4u+1)\big{)}$ $x$'s from \eqref{u-xrelation3}. Then we obtain	
		\begin{align*}
			\sum_{x\in\gf_{p^n}}\chi\big{(}x(x+1)(x^2+x+1)\big{)}=&\sum_{u\in\gf_{p^n}}\chi\big{(}u(u+1))(1+\chi(4u+1)\big{)}\\=&\sum_{u\in\gf_{p^n}}\chi\big{(}u(u+1)\big{)}+\sum_{u\in\gf_{p^n}}\chi\big{(}u(u+1)(4u+1)\big{)}.
		\end{align*}
		Note that $\sum\limits_{u\in\gf_{p^n}}\chi\big{(}u(u+1)\big{)}=-1$ by Lemma \ref{2} and
		\begin{align*}
			\sum_{u\in\gf_{p^n}}\chi\big{(}u(u+1)(4u+1)\big{)}=&\sum_{u\in\gf_{p^n}}\chi\big{(}4u(4u+4)(4u+1)\big{)}\\=&\sum_{v\in\gf_{p^n}}\chi\big{(}(-v-1)(-v+3)(-v)\big{)}\\
			=&\chi(-1)\sum_{v\in\gf_{p^n}}\chi\big{(}v(v+1)(v-3)\big{)}\\
			=&-\lambda^{(1)}_{p,n}.
		\end{align*}
		The desired result follows.

		\item  Clearly, note  that if $p\ne3$ then
		\begin{align*}
			\sum_{x\in \gf_{p^{n}}}\chi\big{(}x(3x^2+2x+3)\big{)}&=\sum_{x\in \gf_{p^{n}}^{*} }\chi\big{(}x(3x^2+2x+3)\big{)}\\&=\sum_{x\in \gf_{p^{n}}^{*} }\chi\big{(}\frac{3x^2+2x+3}{x}\big{)}.
		\end{align*}
		Let $\frac{3x^2+2x+3}{x}=u$, then $x$ and $u$ satisfy
		\begin{equation}\label{u-xrelation4}
			3x^2+(2-u)x+3=0,
		\end{equation}
		which is a quadratic equation on $x$ with the discriminant $\Delta=(u+4)(u-8)$. For each $u\in\gf_{{p^{n}}}$, it corresponds $\big{(}1+\chi(\Delta)\big{)}$ $x$'s from \eqref{u-xrelation4}. Then we obtain
		\begin{align*}
			\sum_{x\in\gf_{p^n}^*}\chi\big{(}\frac{3x^2+2x+3}{x}\big{)}&=\sum_{u\in\gf_{p^n}}\chi(u)\big{(}1+\chi((u+4)(u-8))\big{)}\\&=\sum_{u\in\gf_{p^n}}\chi\big{(}u(u+4)(u-8)\big{)}.
		\end{align*}
		Note that 
		\begin{align*}
			\sum_{u\in\gf_{{p^n}}}\chi\big{(}u(u+4)(u-8)\big{)}&=\sum_{u\in\gf_{{p^n}}}\chi\big{(}\frac{u}{4}(\frac{u}{4}+1)(\frac{u}{4}-2)\big{)}\\&=\sum_{v\in\gf_{{p^n}}}\chi\big{(}v(v+1)(v-2)\big{)}=\lambda^{(2)}_{p,n},
		\end{align*}
		the desired result follows.
		
		\item It is obvious that 
		\begin{align*}
			\sum_{x\in \gf_{p^{n}}}\chi\big{(}(x^2+x+1)(3x^2+2x+3)\big{)}&=\sum_{x^2+x+1\ne0 }\chi\big{(}(x^2+x+1)(3x^2+2x+3)\big{)}\\&=\sum_{x^2+x+1\ne0 }\chi\big{(}\frac{3x^2+2x+3}{x^2+x+1}\big{)}.
		\end{align*}
		Let $\frac{3x^2+2x+3}{x^2+x+1}=u$. It is easy to see that $x=0$ if and only if $u=3$. Moreover, $u$ and $x$ satisfy 
		\begin{equation}\label{u-xrelation5}
			(3-u)x^2+(2-u)x+(3-u)=0.
		\end{equation}
		 When $u\ne3$, \eqref{u-xrelation5} is a quadratic equation on the variable $x$, whose discriminant is $\Delta=-(3u-8)(u-4)$. For each $u\neq 3$, it corresponds $\big{(}1+\chi(\Delta)\big{)}$ $x$'s from \eqref{u-xrelation5}. Then we obtain
		
		\begin{align*}
			\sum_{x^2+x+1\ne0 }\chi\big{(}\frac{3x^2+2x+3}{x^2+x+1}\big{)}=&\sum_{x=0 }\chi\big{(}\frac{3x^2+2x+3}{x^2+x+1}\big{)}+\sum_{u\neq 3}\chi(u)\big{(}1+\chi(-(3u-8)(u-4))\big{)}
			\\=&~\chi(3)+\sum_{u\in\gf_{p^n}}\chi(u)\big{(}1+\chi(-(3u-8)(u-4))\big{)}-2\chi(3)\\=&-\chi(3)-\sum_{u\in\gf_{p^n}}\chi\big{(}u(3u-8)(u-4)\big{)}.
		\end{align*}
		Note that when $p\neq 3$,
		\begin{align*}
			\sum_{u\in \gf_{p^{n}}}\chi\big{(}u(3u-8)(u-4)\big{)}&=\sum_{u\in \gf_{p^{n}}}\chi\big{(}3u(3u-8)(3u-12)\big{)}=\sum_{v\in \gf_{p^{n}}}\chi\big{(}v(v-8)(v-12)\big{)}\\&=\sum_{v\in \gf_{p^{n}}}\chi\big{(}\frac{v}{4}(\frac{v}{4}-2)(\frac{v}{4}-3)\big{)}=\sum_{w\in \gf_{p^{n}}}\chi\big{(}(-w+2)(-w)(-w-1)\big{)}\\&=-\sum_{w\in \gf_{p^{n}}}\chi\big{(}w(w+1)(w-2)\big{)}=-\lambda^{(2)}_{p,n}.
		\end{align*}
		
		\item First observe that
		\begin{align*}
			\sum_{x\in \gf_{p^{n}}}\chi\big{(}x(x^2+x+1)(3x^2+2x+3)\big{)}=\sum_{x\in \gf_{p^{n}}^{*}}\chi\big{(}(\frac{x^2+x+1}{x})(3x^2+2x+3)\big{)}.
		\end{align*}
		Let $\frac{x^2+x+1}{x}=u$, then $x$ and $u$ satisfy
		\begin{equation}\label{u-xrelation6}
			x^2+(1-u)x+1=0,
		\end{equation}
		which is a quadratic equation on $x$ with $\Delta=(u+1)(u-3)$. For each $u\in\gf_{{p^{n}}}$, it corresponds $\big{(}1+\chi(\Delta)\big{)}$ $x$'s from \eqref{u-xrelation6}.
		We mention that $u=-1$ if and only if $x=-1$, for $x\neq0,-1$, we have $(x+1)^2=(u+1)x$, then $\chi(x)=\chi(u+1)$. Moreover, $\chi(3x^2+2x+3)=\chi((3u-1)x)=\chi\big{(}(3u-1)(u+1)\big{)}$.
		We obtain
		\begin{align*}
			\sum_{x\in\gf_{{p^n}}^{*} }\chi\big{(}(\frac{x^2+x+1}{x})(3x^2+2x+3)\big{)}&=\sum_{x=-1 }\chi\big{(}(\frac{x^2+x+1}{x})(3x^2+2x+3)\big{)}\\&~~~~~+\sum_{u\neq -1}\chi\big{(}u(3u-1)(u+1)\big{)}\big{(}1+\chi((u+1)(u-3))\big{)}\\&=-1+\sum_{u\neq -1}\chi\big{(}u(3u-1)(u+1)\big{)}+\sum_{u\neq -1}\chi\big{(}u(3u-1)(u-3)\big{)}\\&=\sum_{u\in\gf_{p^n}}\chi\big{(}u(3u-1)(u+1)\big{)}+\sum_{u\in\gf_{p^n}}\chi\big{(}u(3u-1)(u-3)\big{)}.
		\end{align*}
		Since $p\neq 3$, we have
		\begin{align*}
			\sum\limits_{u\in\gf_{{p^{n}}}}\chi\big{(}u(u+1)(3u-1)\big{)}&=\sum\limits_{u\in\gf_{{p^{n}}}}\chi\big{(}3u(3u+3)(3u-1)\big{)}\\&=\sum\limits_{v\in\gf_{{p^{n}}}}\chi\big{(}(-v)(-v+3)(-v-1)\big{)}\\
			&=-\lambda^{(1)}_{p,n}.
		\end{align*}
		Moreover,
		\begin{align*}
			\sum\limits_{u\in\gf_{{p^{n}}}}\chi\big{(}u(u-3)(3u-1)\big{)}&=\sum\limits_{u\in\gf_{{p^{n}}}}\chi\big{(}3u(3u-9)(3u-1)\big{)}\\&=\sum\limits_{v\in\gf_{{p^{n}}}}\chi\big{(}v(v-9)(v-1)\big{)}\\
			&=\sum\limits_{v\in\gf_{p^{n}}^{*}}\chi\big{(}\frac{(v-9)(v-1)}{v}\big{)}.
		\end{align*} 
		For $v\neq 0$,  let $\frac{v^2-10v+9}{v}=w$, for each $w\in\gf_{p^n}$, it corresponds $\big{(}1+\chi((w+4)(w+16))\big{)}$ $v$'s. Then we obtain
		\begin{align*}
			\sum_{v\in\gf_{p^{n}}^{*}}\chi\big{(}v(v-9)(v-1)\big{)}&=\sum_{w\in\gf_{{p^{n}}}}\chi(w)\big{(}1+\chi((w+4)(w+16))\big{)}\\
			&=\sum_{w\in\gf_{{p^{n}}}}\chi\big{(}w(w+4)(w+16)\big{)}\\&=\sum_{w\in\gf_{{p^{n}}}}\chi\big{(}\frac{w}{4}(\frac{w}{4}+1)(\frac{w}{4}+4)\big{)}\\&=\sum_{x\in\gf_{{p^{n}}}}\chi\big{(}(-x-1)(-x)(-x+3)\big{)}\\&=-\lambda^{(1)}_{p,n}.
		\end{align*}
		We conclude that $\sum\limits_{x\in\gf_{{p^n}}}\chi\big{(}x(3x^2+2x+3)(x^2+x+1)\big{)}=-2\lambda^{(1)}_{p,n}$.
		\end{enumerate}
	\end{proof}
	
	\section{On the number of solutions of certain equation system}\label{B}
	
	This section aimes to determine the  number of solutions  in $(\gf_{{p^n}})^4$ of  the equation system \eqref{equsystem}. To this end, our first objective is to study two related equation systems with solutions in  $(\gf_{p^n}^*)^3$, which will help us achieve our goal.
	\begin{theorem}\label{N(1,1,1)value}
		Let $p^n \equiv 3 (\mathrm{mod}~4)$ and $p\neq 3$. Let $N_{(i,j,k)}$ denote the number of solutions $(y_{1},y_{2},y_{3})\in (\gf_{p^n}^*)^3$ of the equation system
		\begin{equation}\label{N(1,1,1)solution}
			\left\{ \begin{array}{ll}
				y_{1} + y_{2} + y_{3} + 1 &= 0\\
				y_{1}^2+ y_{2}^2+ y_{3}^2 + 1 &= 0
			\end{array} \right.
		\end{equation}
		when $(\chi(y_1), \chi(y_2), \chi(y_3))=(i,j,k)$, $i, j, k\in \{\pm1\}$. We have
		\[N_{(1,1,1)}=\frac{1}{8}\big{(}p^n+3\lambda^{(1)}_{p,n}-6\lambda^{(2)}_{p,n}-15-16\chi(-3)\big{)},\]
		\[N_{(1,1,-1)}=N_{(1,-1,1)}=N_{(-1,1,1)}=N_{(-1,-1,-1)}=\frac{1}{8}\big{(}p^n-3\lambda^{(1)}_{p,n}-3-4\chi(-3)\big{)},\]
		and
		\[N_{(1,-1,-1)}=N_{(-1,1,-1)}=N_{(-1,-1,1)}=\frac{1}{8}(p^n+3\lambda^{(1)}_{p,n}+2\lambda^{(2)}_{p,n}+1),\]
		where $\lambda^{(1)}_{p,n}$ and $\lambda^{(2)}_{p,n}$ are defined in \eqref{sum1} and \eqref{sum2}, respectively.
	\end{theorem}
	
	\begin{proof} 
		First we consider that the number of solutions of \eqref{N(1,1,1)solution}. Note that $y_{1}+y_{3}=-(y_{2}+1)$ and $y_{1}y_{3}=\frac{1}{2}((y_1+y_3)^2-(y^2_1+y^2_3))=y_{2}^2+y_{2}+1$. Moreover, $y_{1}$ and $y_{3}$ satisfy the following quadratic equation on the variable $t$
		\begin{equation}\label{N(1,1,1)simple}
			t^{2}+(y_{2}+1)t+y_{2}^2+y_{2}+1=0.
		\end{equation}
		The discriminant of \eqref{N(1,1,1)simple} is $ \Delta=-3y_{2}^2-2y_{2}-3$. For each given $y_2\in\gf_{p^n}$, we have $(1+\chi(\Delta))$ tuples $(y_1,y_2,y_3)$ that satisfy  \eqref{N(1,1,1)solution}. By Lemma \ref{2}, the number of solutions of \eqref{N(1,1,1)solution} is
		\begin{align*}
			\sum\limits_{y_{2}\in \gf_{{p^n}}}\big{(}1+\chi(\Delta)\big{)}&=p^{n}-\chi(-3).
		\end{align*}
		
		For a solution $(y_{1},y_{2},y_{3})\in (\gf_{p^n})^3 $ of \eqref{N(1,1,1)solution}, we consider the case where there exists some $y_i=0$, where $i\in\{1,2,3\}$.  First, it is easy to see that  $(0,0,0)$  cannot be a solution of \eqref{N(1,1,1)solution}. If there are exactly two zeros in  $(y_{1},y_{2},y_{3})$, then \eqref{N(1,1,1)solution} still has no solution. Remaining  the case where there is only one zero in $ (y_{1},y_{2},y_{3}) $, then without loss of generality, we can assume that $y_{3}=0$,  $y_{1}\neq 0$ and, $y_{2}\neq 0$ and they satisfy 
		\begin{equation}\label{equ2}
			\left\{ \begin{array}{ll}
				{y_1} + {y_2} + 1 &= 0\\
				{y^2_1} + {y^2_2} + 1 &= 0,
			\end{array}\right.
		\end{equation}
		we have $y_2=-y_1-1$ and $ y^2_1+y_1+1=0$. The discriminant of the quadratic equation on $y_1$ is $-3$. Then $ y^2_1+y_1+1=0$ has $\big{(}1+\chi(-3)\big{)}$ solutions, which are not $0$ or $-1$. Since $y_2$ is uniquely determined by $y_1$, \eqref{equ2} has $\big{(}1+\chi(-3)\big{)}$ solutions.  We conclude that \eqref{N(1,1,1)solution} has $ 3\big{(}1+\chi(-3)\big{)}$ solutions containing zeros. 
		
		In the following, we consider $ y_i \ne 0 $ for $ 1 \le i \le 3 $.  Recall that $N_{(i,j,k)}$ denotes the number of solutions $(y_{1},y_{2},y_{3})\in (\gf_{p^n}^*)^3$ of equation system \eqref{N(1,1,1)solution} for $(\chi(y_1), \chi(y_2), \chi(y_3))=(i,j,k)$, $i, j, k\in \{\pm1\}$. It is easy to see that $N_{(1,1,-1)}=N_{(1,-1,1)}=N_{(-1,1,1)}$ and $N_{(1,-1,-1)}=N_{(-1,1,-1)}=N_{(-1,-1,1)}$. Note that $(y_{1},y_{2},y_{3})$ is a solution of \eqref{N(1,1,1)solution} if and only if $(\frac{y_{1}}{y_3}, \frac{y_2}{y_3}, \frac{1}{y_3})$ is a solution of \eqref{N(1,1,1)solution}, then $N_{(-1,-1,-1)}=N_{(1,1,-1)}$. For the convenience, we denote by $N_{(1,1,1)}=\mathcal{N}_{1}$, $N_{(1,1,-1)}=\mathcal{N}_{2}$ and $N_{(1,-1,-1)}=\mathcal{N}_{3}$. 
		Similar discussions  as above, we obtain
		\begin{align}\label{N-sum}
			\mathcal{N}_{1}+4\mathcal{N}_{2}+3\mathcal{N}_{3}=p^n-3-4\chi(-3).
		\end{align}
		
		Next we determine $\mathcal{N}_2$, which is the number of solutions $(y_{1},y_{2},y_{3})$ of \eqref{N(1,1,1)solution} with $(\chi(y_1),\chi(y_2),\chi(y_3))=(1,1,-1)$. Recall that $y_1$ and $y_3$ are two solutions of the quadratic equation \eqref{N(1,1,1)simple}. On one hand, if $(y_1,y_2,y_3)$ is a solution of \eqref{N(1,1,1)simple} with $(\chi(y_1),\chi(y_2),\chi(y_3))=(1,1,-1)$, then $\chi(\Delta)=\chi(-3y_{2}^2-2y_{2}-3)=1$ and $\chi(y_1y_3)=\chi(y_{2}^2+y_{2}+1)=-1$. On the other hand, if there exists some $y_2\in\gf_{p^n}^*$ such that $\chi(y_2)=1$, $\chi(-3y_{2}^2-2y_{2}-3)=1$ and $\chi(y_{2}^2+y_{2}+1)=-1$, then \eqref{N(1,1,1)simple} has two solutions and their product is a nonsquare. More precisely, one of the two solutions is a square element (namely $y_1$), and the other one is a nonsquare element (namely $y_3$). We obtain a unique solution of \eqref{N(1,1,1)solution} by the given $y_2$. We therefore conclude that
		\begin{align*}
			\mathcal{N}_2=&\#\big\{y_{2}\in \gf_{p^n}^*: \chi(y_{2})=1, \chi(-3y_{2}^2-2y_{2}-3)=1,\chi(y_{2}^2+y_{2}+1)=-1\big\}\\
			=&\#\big\{y_{2}\in \gf_{p^n}: \chi(y_{2})=1, \chi(3y_{2}^2+2y_{2}+3)=-1,\chi(y_{2}^2+y_{2}+1)=-1\big\}.
		\end{align*}
		By using the character sums, we get
		\begin{align*}
			\mathcal{N}_{2}=&\frac{1}{8}\sum_{\substack{y_{2}\neq 0,\\3y_{2}^{2}+2y_{2}+3\neq0,\\y_{2}^2+y_{2}+1\ne0}}\big{(}1+\chi(y_2)\big{)}\big{(}1-\chi(3y_{2}^2+2y_{2}+3)\big{)}\big{(}1-\chi(y_{2}^2+y_{2}+1)\big{)}\\={}&\frac{1}{8}\sum_{y_{2}\in\gf_{{p^n}}}\big{(}1+\chi(y_2)\big{)}\big{(}1-\chi(3y_{2}^2+2y_{2}+3)\big{)}\big{(}1-\chi(y_{2}^2+y_{2}+1)\big{)}\\&~~~-\frac{1}{8}\big{(}\sum_{y_{2}=0}+\sum_{3y_{2}^2+2y_{2}+3=0}+\sum_{y_{2}^2+y_{2}+1=0}\big{)}\big{(}1+\chi(y_2)\big{)}\big{(}1-\chi(3y_{2}^2+2y_{2}+3)\big{)}\big{(}1-\chi(y_{2}^2+y_{2}+1)\big{)}.
		\end{align*}

		The above identity holds since any two of the three equations $y_2=0$, $3y_{2}^{2}+2y_{2}+3=0$ and $y_{2}^2+y_{2}+1=0$ cannot hold simultaneously. For each $y_2$ satisfies $3y_{2}^{2}+2y_{2}+3=0$, we have $y_2\neq 0,-1$ and $4y_2=3(y_2+1)^2$. Then $\chi(y_2)=\chi(3)$, consequently $\chi(y_2^2+y_2+1)=\chi\big{(}\frac{1}{3}(3y^2_2+2y_2+3)+\frac{1}{3}y_2\big{)}=\chi(\frac{1}{3}y_2)=\chi(\frac{1}{3})\chi(y_2)=1$. We have $\sum\limits_{3y_{2}^{2}+2y_{2}+3=0}\big{(}1+\chi(y_2)\big{)}\big{(}1-\chi(3y_{2}^2+2y_{2}+3)\big{)}\big{(}1-\chi(y_{2}^2+y_{2}+1)\big{)}=0$.\\
		
		For each $y_2$ satisfies $y_2^2+y_2+1=0$, we can discuss similarly. Note that $y_2^2+y_2+1=0$ has solutions in $\gf_{p^n}$ if and only if $\chi(-3)=1$. For each $y_2$ satisfy $y_2^2+y_2+1=0$, we assert that $y_2\neq 0, -1$ and $y_2=(y_2+1)^2$. Then $\chi(y_2)=1$, consequently $\chi(3y_2^2+2y_2+3)=\chi\big{(}3(y_2^2+y_2+1)-y_2\big{)}=\chi(-y_2)=-1$. The number of such $y_2$ is $\big{(}1+\chi(-3)\big{)}$. We have $\sum\limits_{y_2^2+y_2+1=0}\big{(}1+\chi(y_2)\big{)}\big{(}1-\chi(3y_{2}^2+2y_{2}+3)\big{)}\big{(}1-\chi(y_{2}^2+y_{2}+1)\big{)}=4\big{(}1+\chi(-3)\big{)}$.

		
		By Lemma \ref{2}, Lemma \ref{sums} and the discussions above, we have
		\begin{align*}
			\mathcal{N}_2=&-\frac{1+\chi(-3)}{2}+\frac{1}{8}\sum_{y_2\in\gf_{p^n}}\big{(}1+\chi(y_{2})\big{)}\big{(}1-\chi(3y_{2}^2+2y_{2}+3)\big{)}\big{(}1-\chi(y_{2}^2+y_{2}+1)\big{)}\\=&-\frac{1+\chi(-3)}{2}+\frac{1}{8}\sum_{y\in\gf_{p^n}}\big{[}1+\chi(y_2)-\chi(y_{2}^{2}+y_{2}+1)-\chi(3y_{2}^{2}+2y_{2}+3)-\chi\big{(}y_2(y^2_2+y_2+1)\big{)}\\&~~~-\chi\big{(}y_2(3y_{2}^{2}+2y_{2}+3)\big{)}+\chi\big{(}(y_{2}^{2}+y_{2}+1)(3y_{2}^{2}+2y_{2}+3)\big{)}-\chi\big{(}y_2(y_{2}^{2}+y_{2}+1)(3y_{2}^{2}+2y_{2}+3)\big{)}\big{]}\\{}=&\frac{1}{8}\big{(}p^n-3\lambda^{(1)}_{p,n}-3-4\chi(-3)\big{)}.
		\end{align*}
		
		We can determine the value of $\mathcal{N}_3$ similarly. We know that $\mathcal{N}_3$ is the number of solutions $(y_{1},y_{2},y_{3})$ of \eqref{N(1,1,1)solution} with $(\chi(y_1),\chi(y_2),\chi(y_3))=(1,-1,-1)$. The only difference is $\chi(y_2)=-1$. By Lemmas \ref{2} and \ref{sums}, we have
		\begin{align*}
			\mathcal{N}_3=&\#\big\{y_{2}\in \gf_{p^n}^*: \chi(y_{2})=-1, \chi(-3y_{2}^2-2y_{2}-3)=1,\chi(y_{2}^2+y_{2}+1)=-1\big\}\\
			=&\#\big\{y_{2}\in \gf_{p^n}: \chi(y_{2})=-1, \chi(3y_{2}^2+2y_{2}+3)=-1,\chi(y_{2}^2+y_{2}+1)=-1\big\}\\
			=&\frac{1}{8}\sum_{\substack{y_{2}\neq 0,\\3y_{2}^{2}+2y_{2}+3\neq0,\\y_{2}^2+y_{2}+1\ne0}}\big{(}1-\chi(y_2)\big{)}\big{(}1-\chi(3y_{2}^2+2y_{2}+3)\big{)}\big{(}1-\chi(y_{2}^2+y_{2}+1)\big{)}\\
			=&\frac{1}{8}\big{(}\sum_{y_{2}\in\gf_{{p^n}}}-\sum_{y_{2}=0}-\sum_{3y_{2}^2+2y_{2}+3=0}-\sum_{y_{2}^2+y_{2}+1=0}\big{)}\big{(}1-\chi(y_2)\big{)}\big{(}1-\chi(3y_{2}^2+2y_{2}+3)\big{)}\big{(}1-\chi(y_{2}^2+y_{2}+1)\big{)}\\
			={}&\frac{1}{8}\sum_{y_2\in\gf_{p^n}}\big{(}1-\chi(y_{2})\big{)}\big{(}1-\chi(3y_{2}^2+2y_{2}+3)\big{)}\big{(}1-\chi(y_{2}^2+y_{2}+1)\big{)}\\
			=&\frac{1}{8}\big{(}p^n+3\lambda^{(1)}_{p,n}+2\lambda^{(2)}_{p,n}+1\big{)}.
		\end{align*}
		
		Finally,   $\mathcal{N}_1$  is uniquely determined  from relation \eqref{N-sum} and the values of $\mathcal{N}_2$ and $\mathcal{N}_3$, which completes  the proof.
	\end{proof}
	
		Now we concentrate on determining the number of solutions  in $(\gf_{p^n}^*)^3$ of the following equation system (\ref{n_4equationsystem}), which plays a crucial role in determining the differential spectrum of $x^{\frac{p^n+3}{2}}$ over $\gf_{p^n}$.
	\begin{theorem}\label{n_4value}
		Let $p^n \equiv 3 (\mathrm{mod}~4)$ and $p\neq 3$. Let $n_4$ denote the number of solutions $(y_1,y_2,y_3)\in(\gf_{p^n}^*)^3$ of the equation system
		\begin{equation}\label{n_4equationsystem}
			\left\{ \begin{array}{ll}
				{y_1} + {y_2} + {y_3} + 1 &= 0\\
				{y^d_1} + {y^d_2} + {y^d_3} + 1 &= 0,
			\end{array} \right.
		\end{equation}
		where $d=\frac{p^n+3}{2}$. Then  we have
		\[n_4=\frac{1}{8}\big{(}29p^n-9\lambda^{(1)}_{p,n}-6\lambda^{(2)}_{p,n}-75-32\chi(-3)\big{)},\]
		where $\lambda^{(1)}_{p,n}$ and $\lambda^{(2)}_{p,n}$ are defined in \eqref{sum1} and \eqref{sum2}, respectively.
	\end{theorem}
	\begin{proof}
		Denote by $n_{(i,j,k)}$ the number of solutions $(y_1,y_2,y_3)$ of \eqref{n_4equationsystem} with $(\chi(y_1), \chi(y_2), \chi(y_3))=(i,j,k)$, where $i,j,k\in\{\pm 1\}$. By a similar proof of Theorem \ref{N(1,1,1)value}, we obtain $n_{(1,1,-1)}=n_{(1,-1,1)}=n_{(-1,1,1)}=n_{(-1,-1,-1)}$ and $n_{(1,-1,-1)}=n_{(-1,1,-1)}=n_{(-1,-1,1)}$. Then 
		\[n_{4}=n_{(1,1,1)}+4n_{(1,1,-1)}+3n_{(1,-1,-1)}.\]
		Note that for each $y_{i}\ne 0~ (i=1,2,3),$ we have ${y^d_i}=\chi(y_i)y^{2}_i$. To determine $n_{(1,1,1)}$, $n_{(1,1,-1)}$ and $n_{(1,-1,-1)}$, we shall distinguish three cases and  discuss \eqref{n_4equationsystem} in each case.
		\begin{enumerate}
		\item Case \rm\uppercase\expandafter{\romannumeral1}. ($\chi(y_{1}), \chi(y_2), \chi(y_3))=(1,1,1)$. Then (\ref{n_4equationsystem}) becomes
		\begin{equation*}
			\left\{ \begin{aligned}{}
				{y_1} + {y_2} + {y_3} + 1 &= 0\\
				{y^{2}_1}+ {y^{2}_2} + {y^{2}_3} + 1 &= 0.
			\end{aligned} \right.
		\end{equation*}
		By Theorem \ref{N(1,1,1)value}, we have 
		\[n_{(1,1,1)}=\frac{1}{8}\big{(}p^n+3\lambda^{(1)}-6\lambda^{(2)}_{p,n}-15-16\chi(-3)\big{)}.\]
		
		\item Case \rm\uppercase\expandafter{\romannumeral2}. $(\chi(y_{1}), \chi(y_2), \chi(y_3))=(1,1,-1)$. Then (\ref{n_4equationsystem}) becomes
		\begin{equation}\label{n(1,1,-1)}
			\left\{ \begin{aligned}{}
				{y_1} + {y_2} + {y_3} + 1 &= 0\\
				{y^{2}_1}+ {y^{2}_2} - {y^{2}_3} + 1 &= 0.
			\end{aligned} \right.
		\end{equation}
		Next, we determine $ n_{(1,1,-1)} $. From  (\ref{n(1,1,-1)}), we have ${y_2} + {y_3}=-{y_1}-1$ and ${y^{2}_2} - {y^{2}_3}=-{y^{2}_1}-1$. Note that (\ref{n(1,1,-1)}) has no solution when $ y_{1}=-1 $. Since $y_{1}\ne -1$ and $y_2+y_3\neq 0$, we obtain $y_{2}-y_{3}=\frac{y^2_2-y^2_3}{y_2+y_3}=\frac{y_{1}^{2}+1}{y_{1}+1}$, consequently  $y_{2}=-\frac{y_{1}}{y_{1}+1}$ and $y_{3}=-\frac{y_{1}^{2}+y_{1}+1}{y_{1}+1}$. Note that $y_2$ and $y_3$ are uniquely determined by $y_1$, $(y_{1},y_{2},y_{3})$ is a desired solution with $(\chi(y_{1}), \chi(y_2), \chi(y_3))=(1,1,-1)$ if and only if 
		$\chi(y_1)=1$, $\chi(-\frac{y_{1}}{y_{1}+1})=1$, and $\chi(-\frac{y_{1}^{2}+y_{1}+1}{y_{1}+1})=-1$. We conclude that
		\begin{align*}
			n_{(1,1,-1)}=&\#\big\{y_{1}\in \gf_{{p^n}}: \chi(y_{1})=1, \chi(-\frac{y_{1}}{y_{1}+1})=1,\chi(-\frac{y_{1}^{2}+y_{1}+1}{y_{1}+1})=-1\big\}\\
			=&\#\big\{y_{1}\in \gf_{{p^n}}: \chi(y_{1})=1, \chi(y_{1}+1)=-1,\chi(y_{1}^{2}+y_{1}+1)=-1\big\}\\
			=&\frac{1}{8}\sum_{\substack{y_{1}\neq 0,-1,\\y_{1}^2+y_{1}+1\ne0}}\big{(}1+\chi(y_{1})\big{)}\big{(}1-\chi(y_{1}+1)\big{)}\big{(}1-\chi(y_{1}^{2}+y_{1}+1)\big{)}\\
			=&\frac{1}{8}\big{(}\sum_{y_1\in\gf_{p^n}}-\sum_{y_1=0}-\sum_{y_1=-1}-\sum_{y_{1}^2+y_{1}+1=0}\big{)}\big{(}1+\chi(y_{1})\big{)}\big{(}1-\chi(y_{1}+1)\big{)}\big{(}1-\chi(y_{1}^{2}+y_{1}+1)\big{)}\\
			=&-\frac{1+\chi(-3)}{2}+\frac{1}{8}\sum_{y_1\in\gf_{p^n}}\big{(}1+\chi(y_{1})\big{)}\big{(}1-\chi(y_{1}+1)\big{)}\big{(}1-\chi(y_{1}^{2}+y_{1}+1)\big{)}.\\
		\end{align*}
		The above identities hold since if $y_{1}^{2}+y_{1}+1=0$, then $y_1\neq0, -1$, $y_1$ satisfies $\chi(y_1)=\chi\big{(}(y_1+1)^2\big{)}=1$ and $\chi(y_1+1)=\chi(-y^2_1)=-1$. The number of such $y_1$ is $\big{(}1+\chi(-3)\big{)}$. Hence $\sum\limits_{y_{1}^2+y_{1}+1=0}\big{(}1+\chi(y_{1})\big{)}\big{(}1-\chi(y_{1}+1)\big{)}\big{(}1-\chi(y_{1}^{2}+y_{1}+1)\big{)}=4\big{(}1+\chi(-3)\big{)}$.
		
		By Lemmas \ref{2} and \ref{sums}, we have 
		\begin{align*}
			n_{(1,1,-1)}=&-\frac{1+\chi(-3)}{2}+\frac{1}{8}\sum\limits_{y_{1}\in \gf_{{p^n}}}\big{(}1+\chi(y_{1})\big{)}\big{(}1-\chi(y_{1}+1)\big{)}\big{(}1-\chi(y_{1}^{2}+y_{1}+1)\big{)}\\={}&-\frac{1+\chi(-3)}{2}+\frac{1}{8}\sum\limits_{y_{1}\in \gf_{{p^n}}}\big{[}1+\chi(y_1)-\chi(y_1+1)-\chi(y_{1}^{2}+y_{1}+1)-\chi(y_{1}(y_1+1))\\{}&-\chi\big{(}y_1(y_{1}^2+y_{1}+1)\big{)}+\chi\big{(}(y_{1}+1)(y_{1}^{2}+y_{1}+1)\big{)}+\chi\big{(}y_1(y_{1}+1)(y_{1}^{2}+y_{1}+1)\big{)}\big{]}\\={}&\frac{1}{8}\big{(}p^n-3\lambda^{(1)}_{p,n}-3-4\chi(-3)\big{)},
		\end{align*}
		where $\lambda^{(1)}_{p,n}$ was defined in \eqref{sum1}.
		
		\item Case \rm\uppercase\expandafter{\romannumeral3}. ($\chi(y_{1}), \chi(y_2), \chi(y_3))=(1,-1,-1)$. Then (\ref{n_4equationsystem}) becomes
		\begin{equation}\label{n(1,-1,-1)}
			\left\{ \begin{aligned}{}
				{y_1} + {y_2} + {y_3} + 1 &= 0\\
				{y^{2}_1} - {y^{2}_2} - {y^{2}_3} + 1 &= 0.
			\end{aligned} \right.
		\end{equation}
		If $ y_{3}=-1$, we can obtain $ y_{2}=-y_{1}$, (\ref{n(1,-1,-1)}) has solutions with type $(y_{1},-y_{1},-1)$, we obtain $\frac{p^{n}-1}{2}$ solutions since $\chi(y_1)=1$. Now we assume that $y_{3}\neq -1$, note that ${y_1} + {y_2}=-{y_3}-1\neq 0$ and ${y^{2}_1} - {y^{2}_2}=y^{2}_3-1$, we obtain $y_{1}-y_{2}=-y_{3}+1$, consequently $y_{2}=-1$ and $y_{1}=-y_{3}$. Then  (\ref{n(1,-1,-1)}) has solutions $(-y_{3},-1,y_{3})$ with $\chi(y_3)=-1$. Since $y_3\neq -1$, we obtain $\frac{p^n-3}{2}$ solutions. Hence $n_{(1,-1,-1)}=p^n-2$.

		By the values of $n_{(1,1,1)}$, $n_{(1,1,-1)}$ and $n_{(1,-1,-1)}$, we have
		\begin{align*}
			n_{4}=n_{(1,1,1)}+4n_{(1,1,-1)}+3n_{(1,-1,-1)}=\frac{1}{8}\big{(}29p^n-9\lambda^{(1)}_{p,n}-6\lambda^{(2)}_{p,n}-75-32\chi(-3)\big{)},
		\end{align*}  
which completes the proof.
		\end{enumerate}	
	\end{proof}

	\section{The differential spectrum of the power function $F(x)=x^{\frac{p^n+3}{2}}$ over $\gf_{p^n}$}\label{C}
	
	In this section, we shall focus on studying  the differential spectrum of the power function $F(x)=x^{d}$ over $\gf_{p^n}$, where $d=\frac{p^n+3}{2}$ and $p^n\equiv3\pmod4$. When $p=3$, $F(x)=x^{\frac{3^n+3}{2}}$ is equivalent to $x^{\frac{3^{n-1}+1}{2}}$, whose differential spectrum was determined in \cite{CHNC}. We always assume that $p\neq 3$. Denote by $\delta(b)=\big{|} \{x\in\gf_{p^n}:(x+1)^d-x^d=b\} \big{|}$ for any $b\in\gf_{{p^{n}}}$ and $\omega_i= \big{|} \{b\in \gf_{q}: \delta(b)=i\} \big{|}$. Note that $d$ is an odd integer, then $x$ is a solution of $(x+1)^d-x^d=b$ if and only if $-x-1$ is a solution of $(x+1)^d-x^d=b$. We assert that $\delta(b)$ is an even number except for
	\begin{equation*}
		b=(-\frac{1}{2}+1)^d-(-\frac{1}{2})^d=\frac{\chi(2)}{2}.
	\end{equation*} 
	First, we determine the value of $\delta(\frac{\chi(2)}{2})$ as follows.
	\begin{lemma} 
		Using the notation as above, we have	
		\begin{align*}
			\delta(\frac{\chi(2)}{2})&=\left\{\begin{array}{ll}
				3, & \hbox{if  $\chi(2)=\chi(3)=-1$ or $\chi(2)=\chi(-3)=\chi(\frac{-1+\sqrt{-2}}{2})=-1$,}\\
				1, & \hbox{ otherwise.}
			\end{array}\right.
		\end{align*}
	\end{lemma}
	
	\begin{proof}
		First we assume that $\chi(2)=1$. We consider
		\begin{equation}\label{delta1/2eqn}
			(x+1)^d-x^d=\frac{1}{2}.
		\end{equation}
		It is easy to see that $x=0$ and $x=-1$ are not solutions of (\ref{delta1/2eqn}). For $x\neq0,-1$, $x^d=\chi(x)x^{2}$. We discuss the following four disjoint cases.		\begin{enumerate}
		
			\item Case 1. $(\chi(x+1),\chi(x))=(1,1)$. (\ref{delta1/2eqn}) becomes $2x+\frac{1}{2}=0$, then $x=-\frac{1}{4}$, which contradicts to $\chi(x)=1$. Hence (\ref{delta1/2eqn}) has no solution in this case.
		
		\item Case 2. $(\chi(x+1),\chi(x))=(1,-1)$. (\ref{delta1/2eqn}) becomes $x^2+x+\frac{1}{4}=0$, i.e., $x=-\frac{1}{2}$, which is the unique solution of (\ref{delta1/2eqn}) in this case since $\chi(2)=1$.
		
		\item Case 3. $(\chi(x+1),\chi(x))=(-1,1)$. (\ref{delta1/2eqn}) becomes $x^2+x+\frac{3}{4}=0$, the discriminant of this quadratic equation is $-2$, which is a nonsquare in $\gf_{p^n}$. Then (\ref{delta1/2eqn}) has no solution in this case. 
		
		\item Case 4. $(\chi(x+1),\chi(x))=(-1,-1)$. (\ref{delta1/2eqn}) becomes $2x+\frac{3}{2}=0$, then $x=-\frac{3}{4}$. Consequently, $\chi(x+1)=\chi(\frac{1}{4})=1$, which is a contradiction. Then (\ref{delta1/2eqn}) has no solution in this case. 
	\end{enumerate}
		We conclude that $\delta(\frac{\chi(2)}{2})=1$ when $\chi(2)=1$.
		
Now we assume that $\chi(2)=-1$. Consider
		\begin{equation}\label{delta-1/2eqn}
			(x+1)^d-x^d=-\frac{1}{2}.
		\end{equation}
		Similarly, $x=0$ and $x=-1$ are not solutions of (\ref{delta-1/2eqn}). We discuss in the following four disjoint cases.
		\begin{enumerate}
		\item Case 1. $(\chi(x+1),\chi(x))=(1,1)$. (\ref{delta-1/2eqn}) becomes $2x+\frac{3}{2}=0$, i.e., $x=-\frac{3}{4}$. Then $x+1=\frac{1}{4}$, which is always a square in $\gf_{p^n}$. When $\chi(-3)=1$, we have $\chi(-\frac{3}{4})=1$, (\ref{delta-1/2eqn}) has one solution $x=-\frac{3}{4}$ in this case. When $\chi(-3)=-1$, we have $\chi(-\frac{3}{4})=-1$, (\ref{delta-1/2eqn}) has no solution in this case.
		
		\item Case 2. $(\chi(x+1),\chi(x))=(1,-1)$. (\ref{delta-1/2eqn}) becomes $x^2+x+\frac{3}{4}=0$, then $x(x+1)=-\frac{3}{4}$. (\ref{delta-1/2eqn}) has no solution in this case when $\chi(-3)=1$ since $\chi(x(x+1))=-1$. When $\chi(-3)=-1$, the discriminant of $x^2+x+\frac{3}{4}=0$ is $-2$, the $x^2+x+\frac{3}{4}=0$ has two solutions, namely, $x_1=\frac{-1+\sqrt{-2}}{2}$ and $x_2=\frac{-1-\sqrt{-2}}{2}$, where $\sqrt{-2}$ is a fixed square root of $-2$ in $\gf_{p^n}$. We assert that $\chi(x_1)=\chi(x_2)$ since $\chi(x_1x_2)=\chi(\frac{3}{4})=1$. If $\chi(x_1)=1$, then $\chi(x_2)=1$, (\ref{delta-1/2eqn}) has no solution in this case. If $\chi(x_1)=-1$, then we have $\chi(x_2)=-1$, $\chi(x_1+1)=\chi(-x_2)=1$ and $\chi(x_2+1)=\chi(-x_1)=1$. We obtain two solutions of (\ref{delta-1/2eqn}) in this case.  We conclude that (\ref{delta-1/2eqn}) has two solutions in this case if $\chi(-3)=\chi(\frac{-1+\sqrt{-2}}{2})=-1$ and has no solution otherwise.
		
		\item Case 3. $(\chi(x+1),\chi(x))=(-1,1)$. (\ref{delta-1/2eqn}) becomes $x^2+x+\frac{1}{4}=0$, i.e., $x=-\frac{1}{2}$,  which is the unique solution of (\ref{delta-1/2eqn}) in this case since $\chi(2)=-1$.
		
		\item Case 4. $(\chi(x+1),\chi(x))=(-1,-1)$. (\ref{delta-1/2eqn}) becomes $2x+\frac{1}{2}=0$, i.e., $x=-\frac{1}{4}$. Then $\chi(x)=\chi(-\frac{1}{4})=-1$. Moreover, $x+1=\frac{3}{4}$, When $\chi(-3)=1$, we have $\chi(\frac{3}{4})=-1$, (\ref{delta-1/2eqn}) has one solution $x=-\frac{3}{4}$ in this case. When $\chi(-3)=-1$, we have $\chi(\frac{3}{4})=1$, (\ref{delta-1/2eqn}) has no solution in this case.
	\end{enumerate}
		By discussions as above, when $\chi(2)=-1$, we have $\delta(\frac{\chi(2)}{2})=3$ if $\chi(-3)=1$ or $\chi(-3)=\chi(\frac{-1+\sqrt{-2}}{2})=-1$, the desired result follows.	
	\end{proof}
	
	{\bf Remark.} We give an example to show that there exists $p^n$ such that the former condition  $\delta(\frac{\chi(2)}{2})=3$ holds.  Take  $p^n=59$. Then one gets $\chi(2)=-1$, $\chi(-3)=\chi(\frac{-1+\sqrt{-2}}{2})=-1$, hence $\delta(\frac{\chi(2)}{2})=3$.

	One can immediately deduce the values of $\omega_1$ and $\omega_3$ in the following corollary.
	\begin{corollary}\label{omega}
		With the notation as above, we have 
		\begin{align*}
			\omega_1&=\left\{\begin{array}{ll}
				0, & \hbox{if  $\chi(2)=\chi(3)=-1$ or $\chi(2)=\chi(-3)=\chi(\frac{-1+\sqrt{-2}}{2})=-1$,}\\
				1, & \hbox{ otherwise.}
			\end{array}\right.
		\end{align*}
		and 
		\begin{align*}
			\omega_3&=\left\{\begin{array}{ll}
				1, & \hbox{if  $\chi(2)=\chi(3)=-1$ or $\chi(2)=\chi(-3)=\chi(\frac{-1+\sqrt{-2}}{2})=-1$,}\\
				0, & \hbox{ otherwise.}
			\end{array}\right.
		\end{align*}
	\end{corollary}
	
We now investigate the value of $\delta(1)$. Such value will help determine the number of solutions of a specific equation system (namely, System (\ref{equ3}) involved in the proof of the following main result (Theorem \ref{N_4value})).

		\begin{lemma}\label{delta1}
		With the notation as above, we have $\delta(1)=3+\chi(-3)$.
	\end{lemma}
	
	\begin{proof}
		We consider 
		\begin{equation}\label{delta1eqn}
			(x+1)^d-x^d=1.
		\end{equation}
		It is easy to see that $x=0$ and $x=-1$ are solutions of (\ref{delta1eqn}). For $x\neq0,-1$, we discuss in the following four disjoint cases.
		\begin{enumerate}
			\item Case 1. $(\chi(x+1),\chi(x))=(1,1)$. (\ref{delta1eqn}) becomes $2x+1=1$, i.e., $x=0$, which is a contradiction. Hence (\ref{delta1eqn}) has no solution in this case.
		
		\item Case 2. $(\chi(x+1),\chi(x))=(1,-1)$. (\ref{delta1eqn}) becomes $2(x^2+x)=0$, then $x(x+1)=0$, which is a contradiction. Hence (\ref{delta1eqn}) has no solution in this case.

		\item Case 3. $(\chi(x+1),\chi(x))=(-1,1)$. (\ref{delta1eqn}) becomes $x^2+x+1=0$, the discriminant of this quadratic equation is $-3$. If $\chi(-3)=-1$, then (\ref{delta1eqn}) has no solution in this case. If $\chi(-3)=1$, $x^2+x+1=0$ has two distinct solutions. Moreover, if $x$ satisfies $x^2+x+1=0$, then $x\neq0, -1$,  $\chi(x)=\chi((x+1)^2)=1$ and $\chi(x+1)=\chi(-x^2)=-1$. Then (\ref{delta1eqn}) has two solutions in this case when $\chi(-3)=1$.
		
		\item Case 4. $(\chi(x+1),\chi(x))=(-1,-1)$. (\ref{delta1eqn}) becomes $-2x-1=1$, i.e., $x=-1$, which is a contradiction. Hence (\ref{delta1eqn}) has no solution in this case.
	\end{enumerate}

The desired result follows from the discussions above.
	\end{proof}
	
	To determine the differential spectrum of $F$, it remains to calculate $N_4$, which denotes the number of solutions of the equation system (\ref{equsystem}) given by  Lemma \ref{identicalequation}. Note that $d=\frac{p^n+3}{2}$ is odd when $p^n\equiv3\pmod4$, the number of solutions $(x_{1},x_{2},x_{3},x_{4})\in (\gf_{p^{n}})^{4}$ of the equation system
	\begin{equation}\label{equationsystem}
		\left\{ \begin{array}{ll}
			{x_1} + {x_2} + {x_3} + {x_4} &= 0\\
			{x_1^d} + {x_2^d} + {x_3^d} + {x_4^d} &= 0
		\end{array} \right.
	\end{equation}
	is also $N_{4}$. In the following theorem, we determine the value of $N_4$.
	
	\begin{theorem}\label{N_4value}
		Let $p^n\equiv3\pmod4$ and $p\neq 3$. We have \[{N_4}=1+\frac{1}{8}(p^n-1)(29p^n-9\lambda^{(1)}_{p,n}-6\lambda^{(2)}_{p,n}+5),\] 
		where $\lambda^{(1)}_{p,n}$ and $\lambda^{(2)}_{p,n}$ are defined in \eqref{sum1} and \eqref{sum2}, respectively.
	\end{theorem}
	\begin{proof}
		For a solution $ (x_1,x_2,x_3,x_4) $ of (\ref{equationsystem}), first we consider that there exists $ x_i=0 $ for some $ 0 \le i \le 4 $. It is easy to see that $ (0,0,0,0) $ is a solution of (\ref{equationsystem}), and (\ref{equationsystem}) has no solution containing exactly three zeros. If there are  exactly two zeros in  $ (x_1,x_2,x_3,x_4) $,  then without loss of generality, we  can assume that $x_1=x_2=0$, $x_3,x_4\neq 0$. Hence $x_3=-x_4$. We conclude that (\ref{equationsystem}) has $ 6(p^n-1) $ solutions containing only two zeros. We  assume now that there is exactly one zero in $ (x_1,x_2,x_3,x_4) $. Without loss of generality, we  can assume that $x_4=0$, then $x_1,x_2,x_3\neq 0$ and they satisfy  the following system
		\begin{equation}\label{equ3}
			\left\{ \begin{array}{ll}
				{x_1} + {x_2} + {x_3} &= 0\\
				{x^d_1} + {x^d_2} + {x^d_3} &= 0.
			\end{array}\right.
		\end{equation}
		Let $ y_i=\frac{x_i}{x_3} $ for $ i=1,2 $. Thus, we  have $ y_1+y_2+1=0 $ and $ {y^d_1}+{y^d_2}+1=0 $ with $ y_1,y_2 \ne 0 $. Then $ y_2=y_1+1 $ and $ (y_1+1)^d-y^d_1=1 $. By Lemma \ref{delta1}, we know that equation $ (y_1+1)^d-y^d_1=1 $ has $ 1+\chi(-3) $ solutions in $ \gf_{{p^n}}\setminus\{0,-1\} $. Hence we assert \eqref{equ3} has $4(1+\chi(-3))(p^n-1)$ solutions containing only one zero. We conclude that (\ref{equationsystem}) has $ 1+(10+4\chi(-3))(p^n-1) $ solutions containing zeros. 
		
		Next we consider $ x_i \ne 0 $ for $ 1 \le i \le 4 $. Let $ y_i=\frac{x_i}{x_4} $ for $ i=1,2,3 $. We have
		\begin{equation}\label{n4value}
			\left\{ \begin{array}{ll}
				{y_1} + {y_2} + {y_3} + 1 &= 0\\
				{y^d_1} + {y^d_2} + {y^d_3} + 1 &= 0.
			\end{array} \right.
		\end{equation}
		Denote by $n_4$ the number of solutions  $(y_1,y_2,y_3)\in(\gf_{p^n}^*)^3$ of (\ref{n4value}). By Theorem \ref{n_4value},  the value of $n_4$ is given. Then
		$N_{4}=1+(10+4\chi(-3))(p^n-1)+n_4(p^n-1)$. We complete the proof.
	\end{proof}
	
We are now in a position to determine the differential spectrum of $F(x)=x^{\frac{p^n+3}{2}}$ over $\gf_{p^n}$, which is the main result of the article.
	
	\begin{theorem}\label{differential spectrum}
		Let $F(x)=x^{\frac{p^{n}+3}{2}}$ be the power function over $\gf_{p^{n}}$, where $p^n\equiv3\pmod4$ and $p\neq 3$. The differential spectrum of $F$ is
		\begin{align*}
			DS_{F}=\big\{&\omega_{0}=\frac{1}{64}(37p^n-9\lambda^{(1)}_{p,n}-6\lambda^{(2)}_{p,n}+5),\\&
			\omega_{2}=\frac{1}{32}(11p^n+9\lambda^{(1)}_{p,n}+6\lambda^{(2)}_{p,n}-21),\\&
			\omega_{3}=1,\\&
			\omega_{4}=\frac{1}{64}(5p^n-9\lambda^{(1)}_{p,n}-6\lambda^{(2)}_{p,n}-27)\big\}
		\end{align*}
		when $\chi(2)=\chi(3)=-1$ or $\chi(2)=\chi(-3)=\chi(\frac{-1+\sqrt{-2}}{2})=-1$, and is 
		\begin{align*}
			DS_{F}=\big\{&\omega_{0}=\frac{1}{64}(37p^n-9\lambda^{(1)}_{p,n}-6\lambda^{(2)}_{p,n}-27),\\&
			\omega_{1}=1,\\&
			\omega_{2}=\frac{1}{32}(11p^n+9\lambda^{(1)}_{p,n}+6\lambda^{(2)}_{p,n}-21),\\&
			\omega_{4}=\frac{1}{64}(5p^n-9\lambda^{(1)}_{p,n}-6\lambda^{(2)}_{p,n}+5)\big\}
		\end{align*}
		otherwise, where $\lambda^{(1)}_{p,n}$ and $\lambda^{(2)}_{p,n}$ were defined in \eqref{sum1} and \eqref{sum2}, respectively. Moreover, when $p^n=11$, $F(x)=x^7$ is an APN permutation. For $p^n\neq 11$, the differential uniformity of $F$ is equal to $4$.
	\end{theorem}
	\begin{proof}
		By \eqref{i^2omega}, Theorem \ref{tor} and Theorem \ref{N_4value},  the elements $\omega_{i}$ ($i\in \{0, 1, 2, 3, 4\}$) in the differential spectrum satisfy the following constrain
		\begin{equation}\label{i2omegavalue}
			\sum_{i=0}^{4}i^{2}\cdot\omega_{i}=\frac{1}{8}(21p^n-9\lambda^{(1)}_{p,n}-6\lambda^{(2)}_{p,n}-3).
		\end{equation}
		By Corollary \ref{omega}, \eqref{omegaiomega} and \eqref{i2omegavalue}, $\omega_{0}, \omega_{2}$ and $\omega_{4}$ satisfy 
		\begin{equation*}
			\left\{\begin{array}{ll}
				\omega_{0}+\omega_{2}+\omega_{4}&=p^n-1\\
				2\omega_{2}+4\omega_{4}&=p^{n}-3\\
				4\omega_{2}+16\omega_{4}&=\frac{1}{8}(21p^n-9\lambda^{(1)}_{p,n}-6\lambda^{(2)}_{p,n}-75)
			\end{array}\right.
		\end{equation*}
		when $\chi(2)=\chi(3)=-1$ or $\chi(2)=\chi(-3)=\chi(\frac{-1+\sqrt{-2}}{2})=-1$, and they satisfy 
		\begin{equation*}
			\left\{\begin{array}{ll}
				\omega_{0}+\omega_{2}+\omega_{4}&=p^n-1\\
				2\omega_{2}+4\omega_{4}&=p^{n}-1\\
				4\omega_{2}+16\omega_{4}&=\frac{1}{8}(21p^n-9\lambda^{(1)}_{p,n}-6\lambda^{(2)}_{p,n}-11)
			\end{array}\right.
		\end{equation*}
		otherwise. By solving the above two equation systems, the differential spectrum of $F$ follows. Moreover, by Theorem \ref{bound}, we obtain $|9\lambda^{(1)}_{p,n}+6\lambda^{(2)}_{p,n}|\leq 30p^{\frac{n}{2}}$ and then
		\[\omega_{4}\geq \frac{1}{64}(5p^n-9\lambda^{(1)}_{p,n}-6\lambda^{(2)}_{p,n}-27)\geq \frac{1}{64}(5p^n-30p^{\frac{n}{2}}-27)>0\]
		when $p^n\geq 47$.
		The numerical results show that $\omega_4\geq 1$ for $p^n\in\{7,19,23,31,43\}$, where  $p^n \equiv 3 (\mathrm{mod}~4)$ and $p\neq 3$. We assert that $\omega_4\geq 1$ for all $p^n \equiv 3 (\mathrm{mod}~4)$ , $p\neq 3$ and $p^n\neq 11$.
		This implies that the differential uniformity of such $F$ is $4$. The very specific case $p^n=11$ can be calculated by computing directly the differential spectrum of $F(x)=x^{7}$ over $\gf_{11}$. We find that such spectrum is given by  $\{\omega_{0}=5,\omega_{1}=1,\omega_{2}=5\}$, which completes the proof.
	\end{proof}

Below, we explicit the differential spectrum of $F(x)=x^{\frac{p^n+3}{2}}$ over $\gf_{p^{n}}$ for some specific values of $p$ and $n$.
\begin{example}
	\begin{itemize}
		\item For $p=59$ and $n=1$, the power function $F(x)=x^{31}$ over $\gf_{59}$ is differentially 4-uniform with differential spectrum
		\begin{equation*}
			DS_{F}=\{\omega_{0}=34,\omega_{2}=20,\omega_{3}=1,\omega_{4}=4\}.
		\end{equation*}
		
		\item For $p=11$ and  $ n=3$, the power function $F(x)=x^{667}$ over $\gf_{11^{3}}$ is differentially 4-uniform with differential spectrum
		\begin{equation*}
		   DS_{F}=\{\omega_{0}=785,\omega_{1}=1,\omega_{2}=425,\omega_{4}=120\}.
		\end{equation*}
	
	\item For $p=19$ and  $ n=3$, the power function $F(x)=x^{3431}$ over $\gf_{19^{3}}$ is differentially 4-uniform with differential spectrum
	\begin{equation*}
	     DS_{F}=\{\omega_{0}=3927,\omega_{2}=2434,\omega_{3}=1,\omega_{4}=497\}.
	\end{equation*}
	\end{itemize}
\end{example}	
	
	\section{Concluding remarks}\label{E}

In this paper, we investigated the differential spectrum of the permutation power functions $x^{\frac{p^n+3}{2}}$ over $\gf_{p^n}$ and its related differential uniformity in the case where $p^n \equiv 3 (\mathrm{mod}~4)$ and $p\neq 3$. Our results solve a problem  left open since 1997.  The differential spectrum of $F$ is closely related to two quadratic character sums $\lambda^{(1)}_{p,n}$ and $\lambda^{(2)}_{p,n}$, defined by \eqref{sum1} and \eqref{sum2}, respectively. The character sums can be evaluated by employing ingredients from the theory of elliptic curves over finite fields. The resulting results show that the permutation power functions $F$ have either  excellent or good differential properties by processing a low differential uniformity, namely they are APN permutation when $p^n \equiv 3 (\mathrm{mod}~4)$, $p\neq 3$ and $p^n=11$ and are differentially $4$-uniform permutation when $p^n \equiv 3 (\mathrm{mod}~4)$, $p\neq 3$ and $p^n\neq 11$. We believe  that our techniques could help to determine other cardinalities intervening in similar (or close) problems within this general framework.


\end{document}